\documentclass{article}
\usepackage{ijcai25}

\usepackage{times}
\usepackage{soul}
\usepackage{url}
\usepackage[hidelinks]{hyperref}
\usepackage[utf8]{inputenc}
\usepackage[small]{caption}
\usepackage{graphicx}
\usepackage{amsmath}
\usepackage{amsthm}
\usepackage{booktabs}
\usepackage{algorithm}
\usepackage{algorithmic}
\usepackage[switch]{lineno}

\usepackage{amsfonts}

\newtheorem{theorem}{Theorem}
\newtheorem{lemma}{Lemma}

\title{Streaming Multi-agent Pathfinding}

\author{
Mingkai Tang$^1$
\and
Lu Gan$^1$\And
Kaichen Zhang$^2$\\
\affiliations
$^1$Hong Kong University of Science and Technology\\
$^2$Hong Kong University of Science and Technology (Guangzhou)\\
\emails
\{mtangag, lganaa, kzhangbi\}@connect.ust.hk
}

\begin{document}

\maketitle

\begin{abstract}
The task of the multi-agent pathfinding (MAPF) problem is to navigate a team of agents from their start point to the goal points. However, this setup is unsuitable in the assembly line scenario, which is periodic with a long working hour. To address this issue, the study formalizes the streaming MAPF (S-MAPF) problem, which assumes that the agents in the same agent stream have a periodic start time and share the same action sequence. The proposed solution, Agent Stream Conflict-Based Search (ASCBS), is designed to tackle this problem by incorporating a cyclic vertex/edge constraint to handle conflicts. Additionally, this work explores the potential usage of the disjoint splitting strategy within ASCBS. Experimental results indicate that ASCBS surpasses traditional MAPF solvers in terms of runtime for scenarios with prolonged working hours.
\end{abstract}

\section{Introduction}

The assembly line is a manufacturing process where workpieces are moved between workstations for semi-assembly until final assembly. It is commonly applied in industries such as cars, airplanes, and consumer electronics \cite{boysen2022assembly}. Workpieces can be transferred using a conveyor, which can only move them along fixed paths, making adjustments difficult. An alternative method involves using a team of robots to transport workpieces. The main challenge is navigating robots between workstations without collisions. This challenge is known as the multi-agent pathfinding (MAPF) problem and has been extensively studied. The MAPF problem is applied in various real-world scenarios, including warehouse management \cite{li2021lifelong,xu2022multi,zhang2024guidance}, traffic control \cite{ho2019multi,li2023intersection}, and pipe design \cite{belov2020multi}.

The assembly line operates with a production rhythm. Each workstation periodically receives a workpiece and produces a higher assembly workpiece. The time interval for the period is defined as the cycle time. In the view of the robot team, they need to transport the workpiece between workstations every cycle time. The traditional MAPF problem can be adapted for planning in a periodic scenario with finite working hours, which describes the agent by the start point, the goal point, and the start time. However, in assembly line factories with long working hours, the number of agents grows rapidly, implying that the running time for computing the optimal solution increases exponentially due to the NP-hardness of the MAPF problem \cite{banfi2017intractability}. 
Additionally, in environments where humans and robots share space, the predictability of robot motions is crucial, which cannot be achieved by the traditional MAPF problem setup.

\begin{figure}[t]
    \centering

    \includegraphics[width=0.30\textwidth]{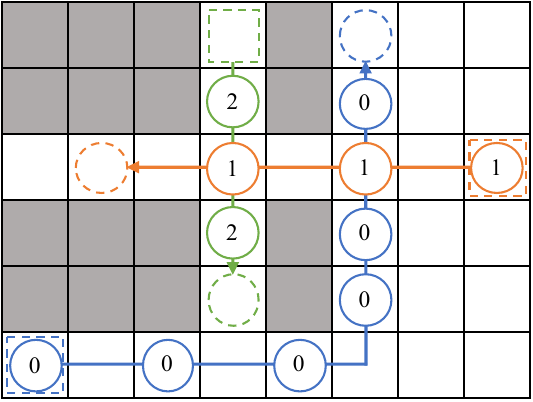}

    \hfill
   
    \caption{A snapshot of the S-MAPF problem where the cycle time is $2$. The snapshot is taken at a time step of $2 \times k$ where $k$ is a sufficiently large integer. White cells are feasible, while grey cells are infeasible. The dashed square marks the start point of the agent stream, and the dashed circle indicates the goal point. The solid circle represents an agent in the stream, with the number denoting the stream ID. Arrows depict the cells that the path crosses. The initial start times for agent streams 0, 1, and 2 are 0, 0, and 1, respectively, with action sequences 'RRRRRUUWUUU', 'LLLLLL', and 'DDDD'. 'U' stands for up, 'D' for down, 'L' for left, 'R' for right, and 'W' for wait. }
    \label{fig:exmaple}
\end{figure}

Hence, we suggest planning with \textit{agent streams} to solve the problems of long working hours and predictability. An agent stream is an agent group containing an infinite number of agents, where agents in an agent stream depart from the start vertex every cycle time and share the same action sequence. At both the start and goal points, there are designated private parking zones for robot loading and unloading. This setup ensures that agents disappear after reaching their goals without disrupting others. It is also adopted in the online MAPF problem \cite{vsvancara2019online}. Formally describing an agent stream, let $s$ be the initial start time of the agent stream and $c$ be the cycle time. Within the agent stream, there exists an agent departing at time step $k \cdot c + s$, where $k \in\mathbb{N}$. We term the task of finding a collision-free path for each agent stream from its start point to its goal point as \textit{streaming multi-agent pathfinding} (S-MAPF). Figure \ref{fig:exmaple} shows an example of the S-MAPF problem. The S-MAPF solution offers endless working hours availability and predictable agent movement when collaborating with humans due to the consistent action sequence. To solve the S-MAPF problem, the Agent Stream Conflict-Based Search (ASCBS) algorithm is proposed, ensuring both optimality and completeness. To resolve the conflict between different agent streams in ASCBS, we introduce the cyclical vertex constraint and the cyclical edge constraint. Additionally, the disjoint splitting strategy \cite{li2019disjoint} is tailored to suit the S-MAPF problem setting. 
The main contributions of this paper are as follows.
\begin{itemize}
    \item We propose the formalization of the S-MAPF problem, which can be adopted in the scenario with unlimited working hours.
    \item We present a two-level approach, ASCBS, to optimally solve the S-MAPF problem. The cyclical vertex constraint and cyclical edge constraint are proposed to resolve agent stream conflicts. We also explore the potential use of disjoint splitting within ASCBS.

    \item We conduct a comprehensive experimental evaluation to compare the computational efficiency of different ASCBS variants. Additionally, an experiment is done that demonstrates ASCBS's high computational efficiency over traditional MAPF solvers in scenarios with long working hours.
    \item We introduce several extensions of the S-MAPF problem that relax certain assumptions.
\end{itemize}

\section{Related Work}
Some MAPF problem variants are introduced to address real-world applications. For example, the MAPF problem with delay probabilities \cite{ma2017multi} accounts for situations where agents may experience delays in their paths. The MAPF problem for large agents \cite{li2019multi} assumes that the agent may occupy more than one vertex. In this study, we introduce the S-MAPF problem tailored to the assembly line scenario, where agents can form streams to continuously transport workpieces. To our knowledge, there are two closely related works for a similar scenario. The Precedence Constrained Multi-Agent Task Assignment and Path-Finding Problem \cite{brown2020optimal} is also applied in assembly, but it focuses on planning for a single final product and emphasizes precedence constraints for operations. The Periodic Multi-Agent Path Planning Problem \cite{kasaura2023periodic} assumes agents' periodic appearances but is defined in continuous space. It is solved as a continuous optimization problem, limiting its applicability to small-scale instances.

In the field of MAPF, conflict-based search \cite{sharon2015conflict} is a well-studied optimal solver, with techniques proposed to accelerate the search process, such as Prioritize Conflicts \cite{boyarski2015icbs}, Disjoint Splitting \cite{li2019disjoint}, and Mutex Reasoning \cite{zhang2022multi}. This study introduces a variant of conflict-based search to solve the S-MAPF problem, incorporating some of these acceleration techniques into the algorithm.
\begin{figure*}[t]
    \centering
    \begin{minipage}{0.285\textwidth}
        \centering
        \includegraphics[width=\textwidth]{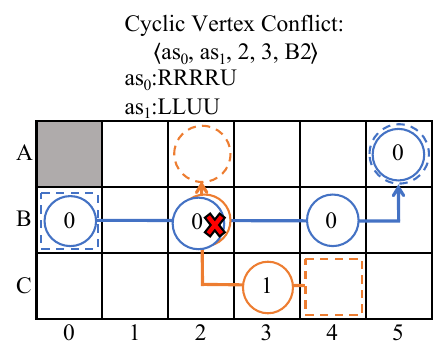}
    \end{minipage}
    \hfill
    \begin{minipage}{0.285\textwidth}
        \centering
        \includegraphics[width=\textwidth]{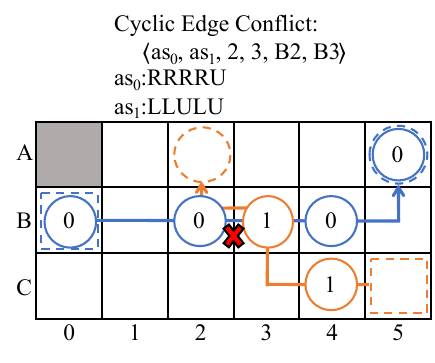}
    \end{minipage}
    \hfill
    \begin{minipage}{0.20\textwidth}
        \centering
        \includegraphics[width=\textwidth]{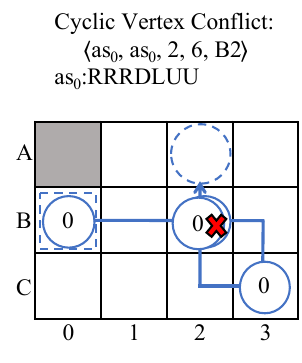}
    \end{minipage}
    \hfill
    \begin{minipage}{0.2\textwidth}
        \centering
        \includegraphics[width=\textwidth]{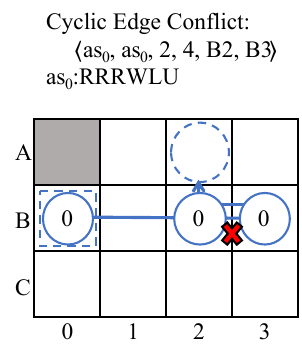}
    \end{minipage}
    \caption{Examples of the conflicts where the cycle time is 2.
    The time step for the snapshot is $2 \times k$, where $k$ is a sufficiently large integer. 
    The initial start times of agent streams $0$ and $1$ are $0$ and $1$. 
    The conflict and the action sequences are shown at the top of the figure.}
    \label{fig:conflict}
\end{figure*}


\section{Problem Definition}
The S-MAPF problem instance is denoted by the tuple $\langle G, c, AS \rangle$. 
In this case, $G = (V, E)$ is an undirected graph with unit edge length, and $c$ denotes the cycle time shared by all agent streams. 
The set $AS = \{as_0, as_1, ..., as_{n-1}\}$ represents the agent streams, where $n$ is the number of agent streams. 
Each $as_i$ can be described as $\{v^s_i, v^g_i, t^s_i\}$ where $v^s_i$ and $v^g_i$ are the start and goal vertices, respectively. 
The initial start time of $as_i$ is denoted by $t^s_i \in [0, c-1]$. 
At each time step $k \cdot c + t^s_i$ for all $k \in \mathbb{N}$, an agent of $as_i$ will appear at $v^s_i$. After that, at each time step, the agent can take an action to move to the neighboring vertex or wait at the current vertex. Agents in the same agent stream are assumed to share the same action sequence, which implies that they also have the same path. In addition, the agent is assumed to disappear after reaching its goal vertex.

The objective of the S-MAPF problem is to find a collision-free path for each agent stream from its start vertex to its goal vertex. The solution can be denoted as $P = \{p_0, p_1,...,p_{n-1}\}$, and the path of $as_i$ is $p_i=[p_i^0, p_i^1, ..., p_i^{l_i-1}]$, with $l_i$ as the path length. We name $p^j_i$ the $j$-th \textit{step} of agent stream $as_i$. For an agent whose start time is $k \cdot c + t^s_i$ for a specific $k$, this agent will be at the vertex $p_i^{t - (k \cdot c + t^s_i)}$ at the time step $t$ when $k \cdot c + t^s_i \leq t < k \cdot c + t^s_i + l_i$. A solution is collision-free if and only if there are no \textit{cyclic vertex conflicts} and \textit{cyclic edge conflicts}. Formally, the cyclic vertex conflict is denoted as $\langle as_i, as_j, q_i, q_j, v \rangle$. It occurs if and only if there exist two nonnegative integers $k_i$, $k_j$ such that the following equations are satisfied:
\begin{equation} \label{equ:v1}
    p_i^{q_i} = v = p_j^{q_j}
\end{equation} \begin{equation} \label{equ:v2}
    (k_i \cdot c + t^s_i) + q_i = (k_j \cdot c + t^s_j) + q_j
\end{equation}
\begin{equation} \label{equ:v3}
    (i - j)^2 + (k_i - k_j)^2 \neq 0
\end{equation}
It means that the agent of $as_i$ with start time $k_i \cdot c + t^s_i$ is located at the same vertex $v$ at the time step $(k_i \cdot c + t^s_i) + q_i$ as the agent of $as_j$ with start time $k_j \cdot c + t^s_j$. In addition, when $as_i$ and $as_j$ are the same agent stream ($j - i = 0$), $k_i$ equal to $k_j$ won't cause a conflict because these two agents are the same agent.
We use $\langle as_i, as_j, q_i, q_j, v_x, v_y \rangle$ to denote the cyclic edge conflict. It occurs if and only if there exist two non-negative integers $k_i$, $k_j$ such that the following equations are satisfied:
\begin{equation} \label{equ:e1}
    p_i^{q_i} = v_x = p_j^{q_{j+1}} \land p_i^{q_{i+1}} = v_y = p_j^{q_j}
\end{equation} 
\begin{equation} \label{equ:e2}
    (k_i \cdot c + t^s_i) + q_i = (k_j \cdot c + t^s_j) + q_j
\end{equation} 
\begin{equation} \label{equ:e3}
    (i - j)^2 + (k_i - k_j)^2 \neq 0
\end{equation}
It implies that the agent of $as_i$ with start time $k_i \cdot c + t^s_i$ crosses the same edge $(v_x, v_y)$ at the time step $(k_i \cdot c + t^s_i) + q_i$ with the agent of $as_j$ with start time $k_j \cdot c + t^s_j$. 
Similar to the cyclic vertex conflict, when $as_i$ and $as_j$ are the same agent stream, $k_i$ equal to $k_j$ won't cause a conflict.
Notably, Equations \ref{equ:v2} and \ref{equ:e2} are the same and can be further derived:
\begin{equation}  \label{equ:ve2}
    t^s_i + q_i \equiv t^s_j + q_j \ (mod \ c)
\end{equation}
It should be noted that $as_i$ and $as_j$ in the cyclic vertex conflict $\langle as_i, as_j, q_i, q_j, v \rangle$ and cyclic edge conflict $\langle as_i, as_j, q_i, q_j, v_x, v_y \rangle$ are not necessarily different, because the agents of the same agent stream might have vertex and edge conflicts. Figure \ref{fig:conflict} presents several examples of the conflicts.

Referencing the MAPF problem, we adopt the sum-of-cost (SOC) for the objective function. Let $C(P)$ denote the SOC of the solution $P$. It can be computed by the equation:
\begin{equation}
    C(P) = \sum_{i=0}^{n-1} (l_i - 1)
\end{equation} 
 An optimal solution is the one that is collision-free and has the lowest SOC compared to other collision-free solutions.

\begin{theorem}
The S-MAPF problem is NP-hard to solve optimally.
\end{theorem}
Theorem 1 can be proven by reduction, and the details are provided in the Supplementary Material.

\begin{algorithm}[tb]
\small
\caption{$ASCBS_{high}$}
\label{alg:high}
\textbf{Input}: agent streams $AS$, graph $G$, cycle time $c$

\begin{algorithmic}[1] 
\STATE  $R$ $\leftarrow$ new node
\STATE  $R.cons$ $\leftarrow$ $\emptyset$
\FOR{ each agent stream $as$ in $AS$ }
\STATE $R.paths[as] \leftarrow$  $ASCBS_{low}(as, R.cons, c)$
\ENDFOR
\STATE $R.cost$ $\leftarrow$ calculate the SOC of $R.paths$.
\STATE $OPEN$ $\leftarrow$ \{$R$\}
\WHILE{$OPEN$ $\neq$ $\emptyset$}
\STATE $N$ $\leftarrow$ minimum cost node from $OPEN$.
\STATE $OPEN$ $\leftarrow$ $OPEN \backslash \{N\}$
\STATE $F$ $\leftarrow$ the best collision in $N$
\IF{$F$ is $None$}
\RETURN $N.paths$
\ENDIF
\STATE $Cons$ $\leftarrow$ build cyclic vertex / edge constraints  or \\  
 \quad \quad \quad  vertex / edge constraints from $F$
\FOR{constraint $con =$  $\overline{(as, v, q_r, q_e)}$ / $\overline{(as, v_i, v_j, q_r, q_e)}$/  \\  \quad \quad \quad $\overline{(as, v, q)}$ / $\overline{(as, v_i, v_j, q)}$ in $Cons$}
\STATE $D$ $\leftarrow$ new node
\STATE $D.cons$ $\leftarrow$ $N.cons$ $\cup$ $con$
\STATE $D.paths[as] \leftarrow$  $ASCBS_{low}(as, D.cons, c)$
\IF{$D.paths[as]$ is not $NULL$}
\STATE $D.cost$ $\leftarrow$ calculate the SOC of $D.paths$
\STATE $OPEN$ $\leftarrow$ $OPEN$ $\cup$ \{$D$\}
\ENDIF
\ENDFOR
\ENDWHILE
\RETURN $NULL$
\end{algorithmic}
\end{algorithm}

\section{Agent Stream Conflict-Based Search}
We present a two-level algorithm, ASCBS, designed to solve the S-MAPF problem. The high-level solver creates a constraint tree (CT) to resolve conflicts from different or the same agent streams. The low-level solver uses the A* algorithm \cite{hart1968formal} to plan the shortest path for an individual agent stream. 
 
\subsection{High-level Solver}
The high-level solver uses CT to avoid conflicts.  Each node in the CT stores the constraint set, the path of each agent stream, and the corresponding SOC. 

The pseudocode of the high-level solver is presented in Algorithm \ref{alg:high}. In lines 1 $\sim$ 4, the low-level solver is invoked at the root node to calculate the shortest path for each agent stream without considering the conflict. During each iteration, the algorithm obtains the CT node with the minimum cost (Line 9) and detects and selects the best conflict (Line 11). The algorithm detects vertex conflicts by Equations \ref{equ:v1}, \ref{equ:ve2}, \ref{equ:v3} and edge conflicts by Equations \ref{equ:e1}, \ref{equ:ve2}, \ref{equ:e3}. The selection of the best conflicts is based on Prioritize Conflicts \cite{boyarski2015icbs}. The prioritization process for cyclic vertex conflicts and cyclic edge conflicts is similar to that for the vertex conflict and the edge conflict. A multi-value decision diagram (MDD) is constructed for each agent stream to evaluate conflict priority. 
By examining the number of occurrences of conflicts at a layer with a width of $1$ in the MDD of two associated agent streams, a conflict with a larger number is given higher priority. 
In Lines 12$\sim$14, if no conflicts are found, the optimal solution is reached, and the algorithm terminates.  In the case of a conflict that involves two distinct agent streams, two cyclic vertex constraints or cyclic edge constraints are generated to resolve the conflict. Alternatively, if the conflict is from the same agent stream, two vertex constraints or edge constraints are established (Line 15). Further details will be explained in the following paragraph. In Lines 16 $\sim$ 24, two child CT nodes are generated, each with one of the two additional constraints. The low-level solver is then utilized to determine the path for the newly constrained agent stream (Line 19). After constructing the child nodes, the algorithm proceeds to the next iteration.

We propose the cyclic vertex constraint and the cyclic edge constraint to resolve the conflict between different agent streams. The cyclic vertex constraint is denoted as $\overline{(as, v, q_r, q_e)}$, indicating that the agent stream $as$ cannot occupy the vertex $v$ at the $q$-th step if $q$ satisfies the equations:
\begin{equation} \label{equ:vcc1}
     q \equiv q_r \ (mod \ c) 
\end{equation}
\begin{equation} \label{equ:vcc2}
     q \neq q_e
\end{equation}
The cyclic edge constraint is denoted as $\overline{(as, v, u, q_r, q_e)}$, meaning that the agent stream $as$ cannot go through the edge from $v$ to $u$ with the $q$-th step if $q$ satisfies the equations:
\begin{equation} \label{equ:ecc1}
     q \equiv q_r \ (mod \ c) 
\end{equation}
\begin{equation}  \label{equ:ecc2}
     q \neq q_e
\end{equation}
Notably, $q_e$ in a cyclic vertex/edge constraint might be $\emptyset$, and in this case, Equations \ref{equ:vcc2} and \ref{equ:ecc2} are always satisfied. 

Here we consider using the constraint to resolve the conflict between different agent streams.
Recall that the cyclic vertex conflict $\langle as_i, as_j, q_i, q_j, v \rangle$ is caused by the ${q_i}$-th step of the path of $as_i$ and the ${q_j}$-th step of the path of $as_j$. When replanning $as_i$'s path, it is not sufficient to just avoid $as_i$ occupying at vertex $v$ at ${q_j}$-th step to resolve the conflict at vertex $v$. This is because the conflict might also occur at $v$ in the $q$-th step of $as_i$'s path where $q \equiv q_i \ (mod \ c)$, according to Equation \ref{equ:ve2}. We should add constraints on all the $q$-th steps that satisfy $q \equiv q_i \ (mod \ c)$. It is symmetric for $as_j$. Specifically, the two child CT nodes are generated respectively as follows to resolve the cyclic vertex conflict $\langle as_i, as_j, q_i, q_j, v \rangle$ where $as_i \neq as_j$:
\begin{itemize}
    \item Add the cyclic vertex constraint $\overline{(as_i, v, q_i, \emptyset)}$ to the constraint set and replan the path of $as_i$.
    \item Add the cyclic vertex constraint $\overline{(as_j, v, q_j, \emptyset)}$ to the constraint set and replan the path of $as_j$.
\end{itemize}
Similarly, to resolve the cyclic edge conflict $\langle as_i, as_j, q_i, q_j, v_i, v_j \rangle$ where $as_i \neq as_j$, the two child CT nodes are operated respectively as follows:
\begin{itemize}
    \item Add the cyclic edge constraint $\overline{(as_i, v_i, v_j, q_i, \emptyset)}$ to the constraint set and replan the path of $as_i$.
    \item Add the cyclic edge constraint $\overline{(as_j, v_j, v_i, q_j, \emptyset)}$ to the constraint set and replan the path of $as_j$.
\end{itemize}

However, when a conflict arises within the same agent stream, we cannot adopt a similar strategy to resolve the conflict using cyclic vertex/edge constraints. This is because the strategy cannot preserve the completeness of the algorithm, as discussed in the Technical Appendix. Instead, we resolve the conflict through the vertex constraint and the edge constraint. The notation $\overline{(as, v, q)}$ represents the vertex constraint, indicating that the agent stream $as$ cannot be located at the vertex $v$ at the $q$-th step. Similarly, the notation $\overline{(as, v, u, q)}$ denotes the edge constraint, which means that the agent stream $as$ cannot go through the edge from the vertex $v$ to the vertex $u$ in the $q$-th step. To resolve the cyclic vertex conflict $\langle as_i, as_j, q_i, q_j , v \rangle$ where $as_i = as_j$, two child CT nodes are generated:
\begin{itemize}
    \item Add the vertex constraint $\overline{(as_i, v, q_i)}$ to the constraint set and replan the path of $as_i$.
    \item Add the vertex constraint $\overline{(as_i, v, q_j)}$ to the constraint set and replan the path of $as_i$.
\end{itemize}
The two child CT nodes are generated to resolve the cyclic edge conflict $\langle as_i, as_j, q_i, q_j, v_i, v_j \rangle$ where $as_i = as_j$.
\begin{itemize}
    \item Add the edge constraint $\overline{(as_i, v_i, v_j, q_i)}$ to the constraint set and replan the path of $as_i$.
    \item Add the edge constraint $\overline{(as_i, v_j, v_i, q_j)}$ to the constraint set and replan the path of $as_i$.
\end{itemize}
\subsection{Low-level solver}
The low-level solver of ASCBS uses an A* algorithm to compute the optimal path of an agent stream under constraints. It is similar to the low-level solver of CBS, where the search progressively explores nodes until reaching the goal vertex. A key distinction is that the ASCBS algorithm's low-level search additionally considers cyclic vertex/edge constraints instead of only vertex/edge constraints. In addition, the Conflict Avoidance Table (CAT) \cite{sharon2015conflict} is employed to help break ties by favoring nodes with fewer conflicts with other agent streams when their $f$ values are equal. Specifically, the nodes use a designated variable to mark the number of collisions along the path from the start vertex to the current vertex. The conflicts with other agent streams at the current step are checked by Equations \ref{equ:v1}, \ref{equ:e1}, \ref{equ:ve2}. Notably, conflicts within the same agent stream are not counted because it needs to construct the entire path from the start vertex to the current vertex, significantly increasing the runtime.

There is an alternative implementation in that the low-level solver uses the Iterative-Deepening-A* (IDA*) algorithm \cite{korf1985depth} to generate the path. The IDA* algorithm utilizes a depth-first search approach, which allows it to maintain the entire path to the current search state. This capability enables the pruning of states when a cyclic vertex/edge conflict arises from the same agent stream during the search. Consequently, the IDA* algorithm can ensure that no collision exists from the same agent stream in the result path.
As a result, the high-level solver only needs to resolve the conflict between agents of different agent streams when adopting the IDA* algorithm as the low-level solver. In contrast, the A* algorithm is unable to prune states with a cyclic vertex/edge conflict during the search because it does not maintain the entire path at every state. Such conflicts can only be resolved by the high-level solver.

\subsection{Disjoint Splitting} \label{subsec:disjoint}
Referencing the improvement technique of the CBS algorithm, we consider adopting the disjoint splitting strategy \cite{li2019disjoint} into the ASCBS algorithm. This strategy involves applying asymmetric constraints to two child CT nodes to divide the problem into two separate subproblems, which helps prevent redundant searches.

In this context, $(as, v, q)$ represents a positive vertex constraint, indicating that the agent stream $as$ must be placed at the vertex $v$ in the $q$-th step. Similarly, $(as, v, u, q)$ denotes a positive edge constraint, indicating that the agent stream $as$ must pass through the edge from vertex $v$ to vertex $u$ at the $q$-th step. 
The positive vertex/edge constraint can impact other agent streams' paths as well as other steps of the current agent stream's path. For example, when $(as_i, v, q_i)$ is added into the constraint set, it restricts another agent stream $as_j$ from being at the vertex $v$ in the ${q_j}$-th step, where $t^s_i + q_i \equiv t^s_j + q_j \ (mod \ c)$. 
Furthermore, the path of $as_i$ cannot be at vertex $v$ at ${q_i}'$-th step where $q^i \equiv {q_i}' \ (mod \ c)$ and $q^i \neq {q_i}'$. 

Specifically, the two child CT nodes are generated to resolve the cyclic vertex conflict $\langle as_i, as_j, q_i, q_j, v \rangle$. Let $k$ represent the agent stream ID for the positive constraint. 
Let $k'$ denote the ID of the newly constrained agent stream in the other child CT node. Formally, If $as_i \neq as_j$, then $k' = \{i, j\} \backslash k$; otherwise, $k' = k$.
\begin{itemize}
    \item Add the positive vertex constraint $(as_k, v, q_k)$, the cyclic vertex constraints $\overline{(as_k, v, q_k, q_k)}$ and $\{\overline{(as_o, v, t^s_k + q_k - t^s_o, \emptyset)}: o \neq k\}$ to the constraint set and replan the path of $as_k'$.
    \item Add the vertex constraint $\overline{(as_k, v, q_k)}$ to the constraint set and replan the path of $as_k$.
\end{itemize}
To resolve the cyclic edge conflict $\langle as_i, as_j, q_i, q_j, v_i, v_j \rangle$, two child CT nodes are generated.
\begin{itemize}
    \item Add the positive edge constraint $(as_k, v_i, v_j, q_k)$, the cyclic edge constraints $\overline{(as_k, v_j, v_i, q_k, q_k)}$ and $\{\overline{(as_o, v_j, v_i, t^s_k + q_k - t^s_o, \emptyset)}: o \neq k\}$ to the constraint set and replan the path of $as_k'$.
    \item Add edge constraint $\overline{(as_k, v_i, v_j, q_k)}$ to the constraint set and replan the path of $as_k$.
\end{itemize}
For the selection of $k$, we randomly choose $k$ from $\{i, j\}$ if $as_i \neq as_j$. Otherwise, $k \leftarrow \arg \min_{o\in\{i,j\}}(q_o)$, where the motivation is that the later step of the path depends on the earlier step of the same path, and it is unreasonable to keep still the agent stream in the later step while altering it in the earlier step.

\begin{figure*}[t]
    \centering
    \includegraphics[width=0.95\textwidth]{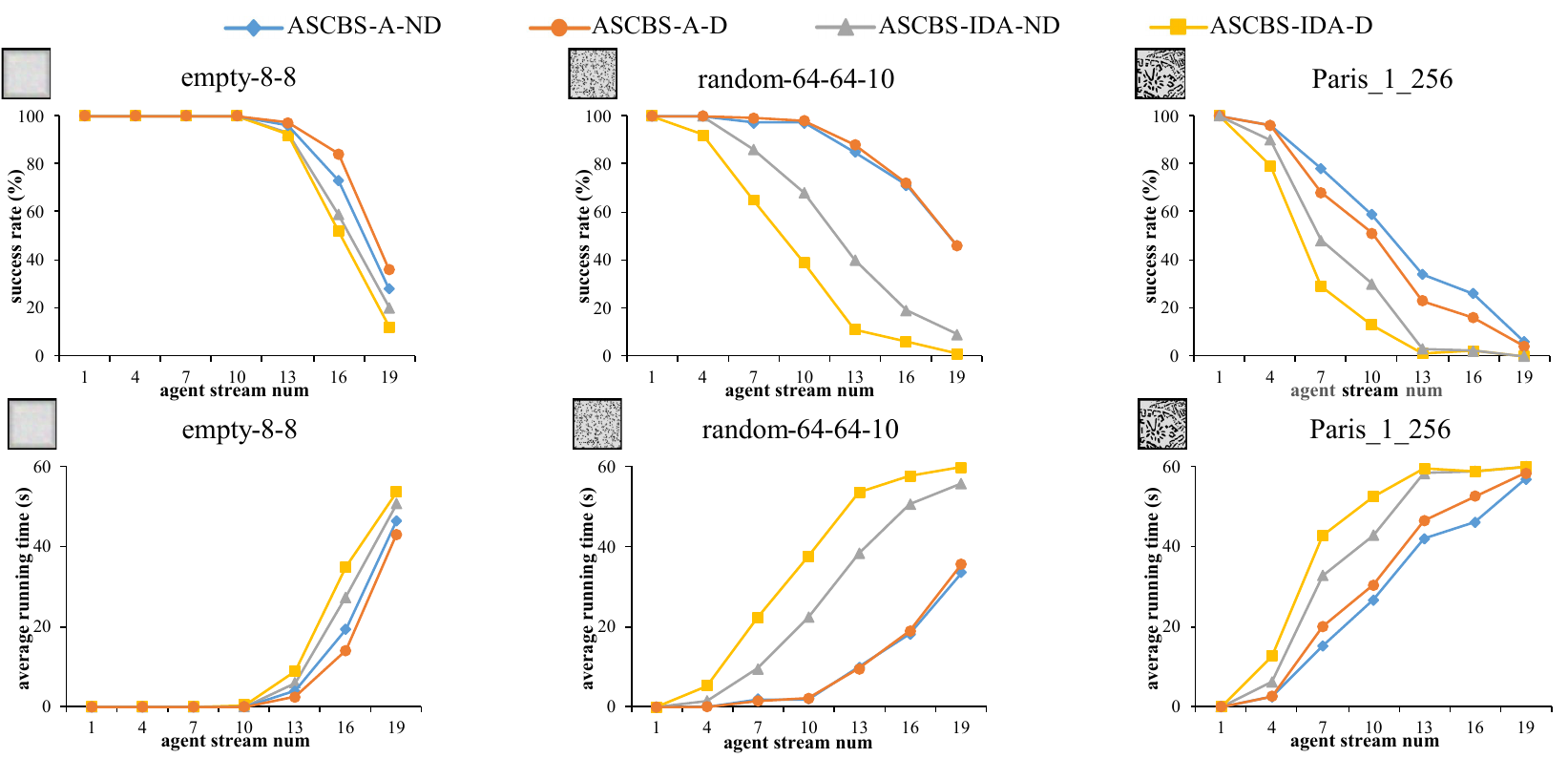}
    \caption{Success rate and average running time for different implementations of ASCBS on the instances with cycle time $3$.}
    \label{fig:exp1}
\end{figure*}

\begin{figure*}[t]
    \centering
    \includegraphics[width=0.95\textwidth]{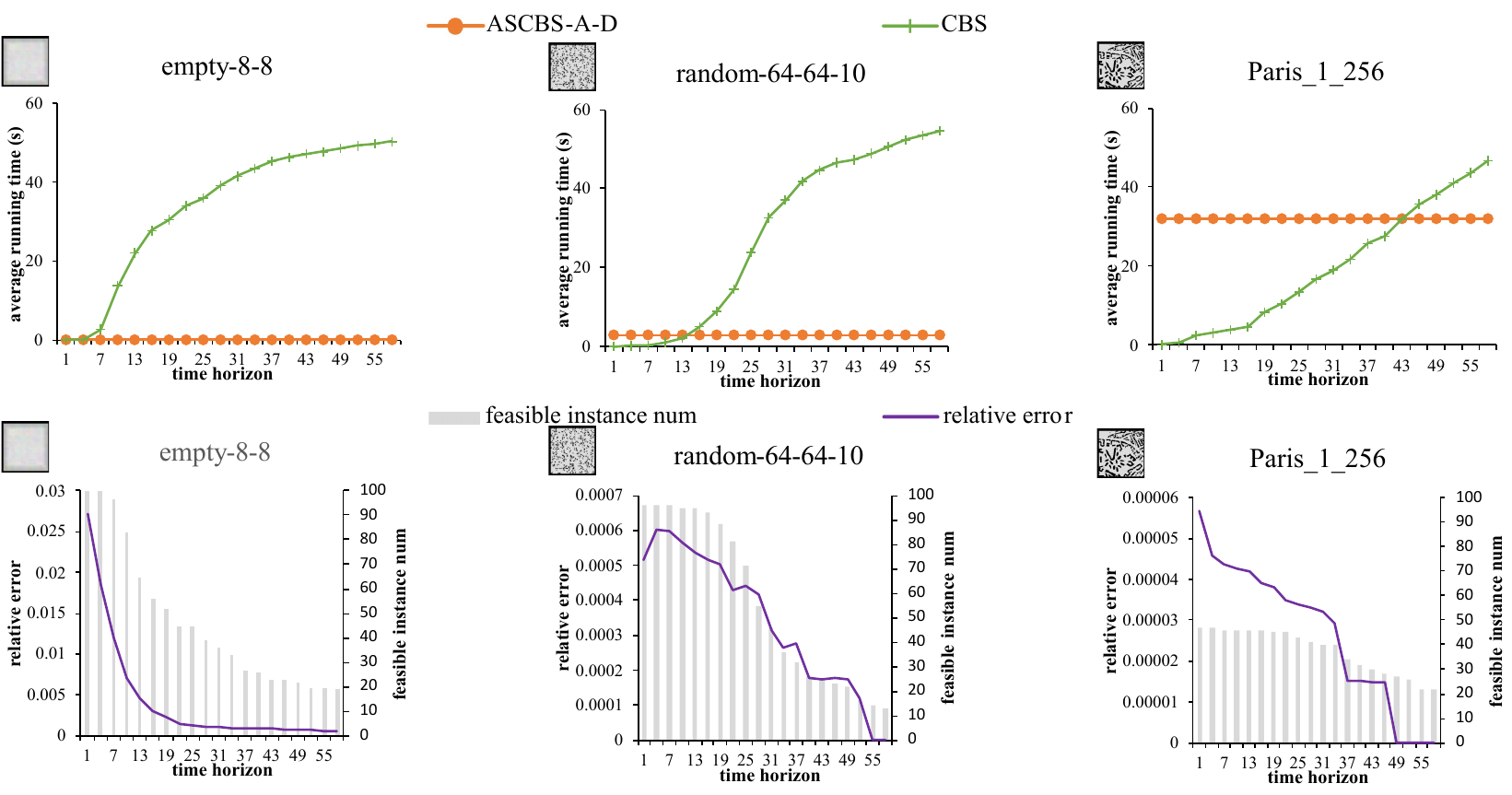}
    \caption{Average running time of ASCBS and CBS, the number of feasible instances by both two algorithms, and average relative error of ASCBS compared to CBS on the instance with $10$ agent streams and cycle time $3$.}
    \label{fig:exp3}
\end{figure*}

\subsection{Theoretical Analysis}

Based on the implementation of the ASCBS, the following theorem holds:
\begin{theorem} \label{theo:2}
The ASCBS algorithm that employs A* as the low-level solver is an optimal and complete solution for the S-MAPF problem.
\end{theorem}
The proof of Theorem \ref{theo:2} is provided in the Supplementary Material.

\section{Experiment} \label{sec:experiment}
 This section presents experiments conducted to evaluate the computational efficiency and solution quality of the ASCBS on three grid-like graphs selected from the MAPF Benchmark Set  \cite{stern2019mapf}, with varying scaling and features.  
The three chosen graphs are 'empty-8-8' (size: 8 $\times$ 8), 'random-64-64-10' (size: 64 $\times$ 64), and  'Paris\_1\_256' (size: 256 $\times$ 256). 
Each graph includes 25 scenario files that specify the start and goal points of the agents. For every scenario, 4 instances are generated with a randomly assigned initial start time for each agent stream within the range of $[0, c-1]$ for the setting of circle time $c$ and agent number $n$, resulting in a total of 100 instances for each setting. All experiments are performed on a Ubuntu 20.04 computer equipped with an Intel Core i7-8700 CPU running at 3.2 GHz with 32GB of main memory. The code is publicly available at https://github.com/tangmingkai/S-MAPF.

The following algorithms are used in the experiments:
\begin{itemize}
    \item \textbf{ASCBS-A-ND}: ASCBS with A* as the low-level solver and without the disjoint splitting strategy. 
    \item \textbf{ASCBS-A-D}: ASCBS with A* as the low-level solver and adopting the disjoint splitting strategy. 
    \item \textbf{ASCBS-IDA-ND}: ASCBS with IDA* as the low-level solver and without the disjoint splitting strategy. 
    \item \textbf{ASCBS-IDA-D}: ASCBS with IDA* as the low-level solver and adopting the disjoint splitting strategy.
     \item \textbf{CBS}: The conflict-based search with WDG guidance \cite{li2019improved} and mutex reasoning \cite{zhang2022multi}, with the setting that agents disappear upon reaching their goals. It is the solver for the traditional MAPF problem, but ignore the assumption that agents in the same agent stream must share the same action sequence.
\end{itemize}

\subsection{Various Number of Agent Streams}
We assessed the computational efficiency of the different implementations of ASCBS on the instances with varying numbers of agent streams, where the circle time is set at $3$. The average running time and the success rate are used as evaluation metrics. An instance is considered unsuccessful if the running time exceeds 60 seconds, in which case the running time is set at 60 seconds directly.

The results are presented in Figure \ref{fig:exp1}. The performance of ASCBS-A-ND and ASCBS-A-D surpasses that of ASCBS-IDA-ND and ASCBS-IDA-D. 
This difference in performance arises because the IDA* algorithm repeatedly executes the search from the start vertex of an agent stream whenever the bound of the f-value increases. The overhead associated with using the IDA* algorithm as the low-level solver outweighs its advantage, which is that the high-level solver does not need to resolve conflicts from the same agent stream. This finding suggests that resolving conflicts from the same agent streams at the high-level solver is more efficient than resolving them at the low-level solver.
In the empty-8-8 and random-64-64-10, ASCBS-A-D outperforms ASCBS-A-ND, with the gap more pronounced in smaller maps. In contrast, on the larger map (Paris\_1\_256), ASCBS-A-ND outperforms ASCBS-A-D. We believe that the reason is as follows. On the one hand, the strategy that imposes a larger additional constraint set on child CT nodes has a higher efficiency because it is more likely to elevate the lower bound of the optimal SOC of the CT node. A cyclic vertex/edge constraint on the agent stream with path length $l$ and cycle time $c$ can be extracted to approximately $\lfloor \frac{l}{c} \rfloor$ vertex/edge constraint. In the solver with the disjoint splitting strategy, one of the child CT nodes has an additional constraint set that only contains a vertex/edge constraint, which is smaller than that in the solver without disjoint splitting. This makes the non-disjoint splitting strategy more beneficial than the disjoint splitting strategy in terms of computational efficiency. Consequently, the advantage caused by non-disjoint splitting diminishes in a smaller size map due to the shorter path length $l$, resulting in a smaller value $\lfloor \frac{l}{c} \rfloor$.
On the other hand, the disjoint splitting prevents the same conflict recurrence in the subtree rooted by two child CT nodes.
In this view, the disjoint splitting strategy is more beneficial than the non-disjoint splitting strategy. Considering the above two factors, when the map size is small, the advantage caused by the disjoint splitting from the solver with A* search is larger than the disadvantage. Nonetheless, ASCBS-IDA-D is outperformed by ASCBS-IDA-ND, as the computing time of IDA* constitutes a significant portion of the total computing time and is heavily influenced by constraints. Therefore, the advantage of the disjoint splitting strategy in the solver with IDA* cannot cover the disadvantage. 

\subsection{Comparing ASCBS with CBS}
We evaluate the computational efficiency and quality of the solution of ASCBS-A-ND compared to CBS. Since CBS is unable to generate solutions for unlimited working hours, we introduce an integer variable called the time horizon. This variable restricts CBS to considering only the agents whose start time is less than or equal to the time horizon. Our evaluation includes instances with $10$ agent streams and circle time $3$. We measure computational efficiency using average running time. We named the instance feasible by the algorithm if its running time is within 60 seconds and can produce a feasible solution. To evaluate solution quality, we calculate the average relative error of the SOC among instances feasible by both algorithms. Notably, the traditional MAPF problem has less constraint than the S-MAPF problem, resulting in a no larger SOC in CBS than in ASCBS.

Figure \ref{fig:exp3} displays the average running time of ASCBS and CBS, the number of instances feasible by both algorithms and the average relative error of ASCBS compared to CBS. It is noteworthy that the solution of ASCBS can be applied across all time horizons, leading to consistent average running times. As the time horizon increases, CBS's running time scales due to the larger number of agents considered. In contrast, the relative error of ASCBS compared to CBS is very small and decreases as the time horizon grows. This trend suggests that when the working hours are extended, the difference in solution quality between ASCBS and CBS diminishes, supporting the assumption that ASCBS is more suitable to apply in the scenario with long working hours compared to CBS.

\section{Extensions}
The formalization of the S-MAPF problem is based on several strong assumptions. In this section, we introduce extensions of the S-MAPF problem to relax these assumptions. It is evident that the solver for the basic S-MAPF problem can be adapted to solve the extended problem. We denote the basic S-MAPF problem as $P1$.

\subsection{Stay in the Environment}
The S-MAPF problem assumes that agents in the agent stream appear at the starting vertex at specific time steps and disappear upon reaching the goal vertex. This assumption is suitable for the intersection scenarios and the environments that have some private parking zones at both the start and goal vertices for robot loading and unloading \cite{vsvancara2019online}. An extension of the S-MAPF problem (denoted as $P2$) relaxes this assumption, by requiring all agents to remain in the environment throughout the entire working process.

Several methods can be employed to solve the extended problem. One method is to generate a return path for each agent stream, dividing the $P2$ into two stages. In the first stage, a path is generated for each agent stream from its start vertex to its goal vertex, similar to $P1$. The second stage focuses on generating a path for each agent stream from the goal vertex back to the start vertex, ensuring that this path avoids conflicts with the paths established in the first stage. Additionally, the path length in the second stage must adhere to a specific constraint to maintain the cyclic pattern. Let $l_1^i$ and $l_2^i$ represent the length of the path generated in the first and second stages for agent stream $i$, receptively. These path lengths must satisfy the condition $\forall i, l_1^i + l_2^i \equiv 0 \, \, (mod \, \, c)$. In the implementation, the ASCBS algorithm can enhance the goal-reaching conditions in the low-level solver to ensure that this condition is met.

An alternative method for the $P2$ is to jointly generate the paths in both stages, ensuring that the path lengths are multiples of $c$. This can be achieved by employing a multi-label A* algorithm \cite{grenouilleau2019multi} as the low-level solver within the ASCBS algorithm. 

A more advanced method for the $P2$ is to allow for the concentration of multiple agent streams. In this method, an agent that has reached the goal vertex of an agent stream can subsequently move to the start vertex of another agent stream. Future research will focus on this category of methods.


\subsection{Non-Uniform Cycle Time}
In the context of $P1$, we assume that all agent streams operate on the same cycle time. This assumption restricts the applicability of the S-MAPF problem to scenarios that require non-uniform cycle times. In the case of the problem with non-uniform cycle time agent streams (denoted as $P3$), conflicts also display in a periodic pattern, similar to that observed in $P1$. Consequently, we can extend the definition of cyclic conflict and constraints of the ASCBS algorithm to accommodate non-uniform cycle times. Further details are available in the Supplementary Material.

\section{Conclusion}
This work formalized the S-MAPF problem, which assumes that agents in the same agent stream have a periodic start time and use the same action sequence. To solve the S-MAPF problem, an optimal and complete algorithm, ASCBS, was introduced, which includes cyclic vertex/edge constraints. The potential of the disjoint splitting strategy in the ASCBS algorithm was also investigated. Experimental comparisons were conducted to evaluate the performance of different implementations of ASCBS, indicating that resolving conflicts within the same agent stream in the high-level solver is more effective than in the low-level solver, and disjoint splitting is advantageous when the map size is small. Moreover, experiments demonstrated that when the working hour is long, ASCBS outperforms CBS in terms of runtime while maintaining a slightly lower or equal solution quality. Finally, we presented several extensions of the S-MAPF problem. The solver developed for the basic S-MAPF problem can be adapted to solve these extended problems.

\bibliographystyle{named}
\bibliography{ijcai25}
\appendix
\setcounter{theorem}{0}

\begin{figure*}[!t]
    \centering
    \includegraphics[width=0.98\textwidth]{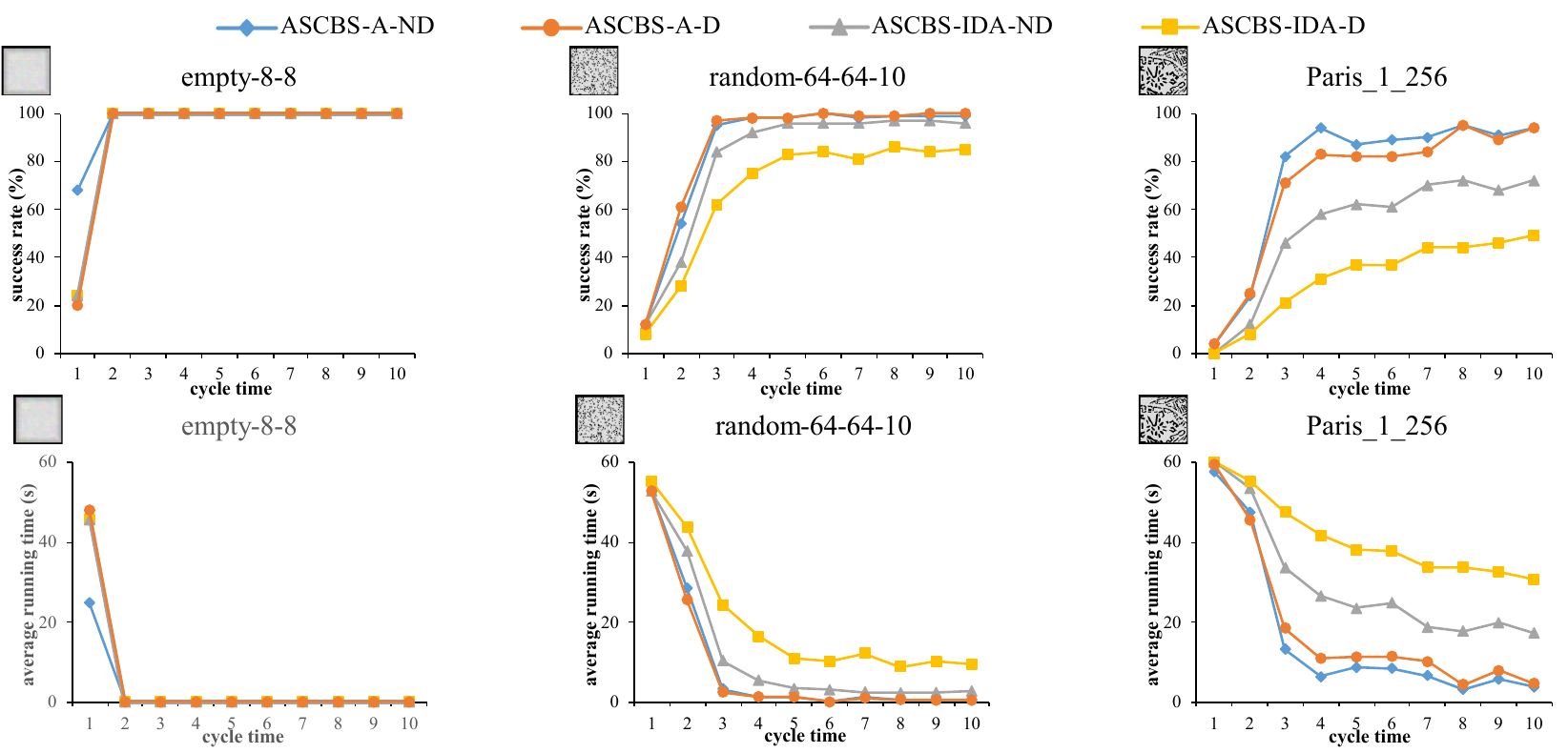}
    \caption{Success rate and average running time for different implementations of ASCBS on the instances with $7$ agent stream.}
    \label{fig:exp2}
\end{figure*}

\section{Theoretical Analysis}
\subsection{Streaming Multi-agent Pathfinding (S-MAPF)}
This subsection provides a theoretical analysis of the S-MAPF problem.
\begin{lemma} \label{lemma:mapf}
The MAPF problem, in which agents disappear upon reaching their goals, is NP-hard to solve optimally.  
\end{lemma}
\begin{proof}
The proof of NP-hardness for optimally solving the  MAPF problem is provided by [Banfi et al., 2017].In this proof, a reduction from the satisfiability problem (SAT) is employed to create MAPF instances, where all candidate optimal solutions do not require passing through the goals of other agents. Therefore, the setting related to reaching the goal is independent of the proof. This indicates that the NP-hardness of the MAPF problem is independent of whether agents stay at the goal or disappear.
\end{proof}
\begin{theorem}
The S-MAPF problem is NP-hard to solve optimally.
\end{theorem}
\begin{proof}
The MAPF problem can be reduced to the S-MAPF problem through the following steps.
Initially, a feasible solution of any solvable MAPF problem instance can be calculated in a polynomial-time complexity, and its SOC is bounded by $O(|V|^3)$ \cite{kornhauser1984coordinating}. Therefore, the makespan ($max_{i=0}^{n-1}(l_i-1))$ is also bounded by $O(|V|^3)$, where $l_i$ is the path length of agent $i$. Let $w$ represent the makespan of that feasible solution. Subsequently, an S-MAPF instance can be created by using the same graph as the MAPF instance, and the cycle time is set to be an integer larger than $w$. The start and goal vertices of the agent stream in the S-MAPF instance are the same as those of the agent in the MAPF instance, with the initial start time set to $0$. Consequently, the optimal path of the S-MAPF problem is the same as the optimal path of the MAPF problem because of the sufficiently large cycle time, ensuring that each agent completes its path before the subsequent agent appears, preventing interference. 
Since optimally solving the MAPF problem is recognized as NP-hard, as stated in Lemma \ref{lemma:mapf}, the NP-hardness of achieving optimal solutions for the S-MAPF problem is established. \end{proof}

\subsection{Agent Stream Conflict-Based Search (ASCBS)}
This subsection will prove the optimality and completeness of the Agent Stream Conflict-Based Search (ASCBS) algorithm. In the proof of the following two lemmas, we provide only the proof of the cyclic vertex conflict, while the proof of the cyclic edge conflict is omitted as it is symmetric. Let $\langle as_i, as_j, q_i, q_j, v \rangle$ represent the current cyclic vertex conflict.

\begin{lemma} \label{lemma:1}
The newly added constraints can resolve the corresponding conflict in the constraint tree (CT) node.
\end{lemma}
\begin{proof}
Now we consider the ASCBS algorithm without the disjoint splitting strategy. When $as_i \neq as_j$, due to the cyclic vertex constraint $\overline{(as, v, q, \emptyset)}$ being a superset of the vertex constraint $\overline{(as_i, v, q_i)}$, the replanned path of the agent stream $as_i$ avoids being located at vertex $v$ at $q_i$-th step to resolve the conflict. So as the agent stream $as_j$.
In the case where $as_i = as_j$, either the $q_i$-th step or the $q_j$-th step, driving from vertex $v$, can avoid the conflict. 

Considering the ASCBS algorithm with the disjoint splitting strategy, one of the child CT nodes employs a positive constraint to force the agent stream  to occupy the vertex $v$ at the corresponding step and drives away the path that might conflict with this constraint. The other child CT node avoids the agent stream being located at vertex $v$ at the corresponding step to resolve the conflict.
\end{proof}

\begin{lemma} \label{lemma:2}
A collision-free path satisfying the constraint set of a CT node must also satisfy the constraint set of at least one of its child CT nodes. 
\end{lemma}
\begin{proof}

In the ASCBS algorithm without the disjoint splitting strategy, we prove the lemma by contradiction in the case of $as_i \neq as_j$. If the lemma is false, then there exists a collision-free path that satisfies the constraint set of the parent CT node but does not satisfy both of the cyclic vertex constraints $\overline{(as_i, v, q_i, \emptyset)}$ and $\overline{(as_j, v, q_j, \emptyset)}$. Let $q_i'$ be such that $q_i' \equiv q_i \ (mod \ c)$ and the path of the agent stream $as_i$ will be located at $v$ in the ${q_i'}$-th step. Similarly, let $q_j'$ be such that $q_j' \equiv q_j \ (mod \ c)$ and the path of the agent stream $as_j$ will be located at $v$ at ${q_j'}$-th step. According to Equations 1, 3, and 7 in the main text, the cyclic vertex conflict $\langle as_i, as_j, q_i', q_j', v \rangle$ occurs, making the path not conflict-free. Thus, the lemma is proven in this case. The proof for the case of $as_i = as_j$, where two child CT nodes are generated using the vertex, follows the same logic as in the CBS algorithm.

Let $k$ and $k'$ have the same meaning as in Subsection 4.3 in the main text. 
In the ASCBS algorithm with the disjoint splitting strategy, if the lemma is false, then the agent stream $as_k$ is at the vertex $v$ in the $q_k$-th step. Furthermore, one of the following two situations must hold true: The first situation is that there exists a $q_k' \neq q_k$ such that $q_k' \equiv q_k \ (mod \ c)$ and the path of the agent stream $k$ is located in $v$ in the ${q_k}$-th step. The second situation is that there exists an agent stream $as_o$ and $q_o$ such that $q_o+t^s_o \equiv q^s_k + q_k\ (mod \ c)$, $o \neq k$, and the path of agent stream $as_o$ is located at $v$ at ${q_o}$-th step. Both situations lead to a conflict with the agent stream $as_k$, resulting in the path being non-collision-free.


\end{proof}

\begin{theorem}
The ASCBS algorithm that employs A* as the low-level solver is an optimal and complete solution for the S-MAPF problem.
\end{theorem}
\begin{proof}
    By combining Lemmas \ref{lemma:1} and \ref{lemma:2} with the optimality and completeness of the A* algorithm and the CBS algorithm, the optimality and completeness of the ASCBS are proven. This proof covers the ASCBS algorithm with and without the disjoint splitting strategy, according to Lemmas \ref{lemma:1} and \ref{lemma:2}.
\end{proof}

\section{Further Experimental Result}
We assess the computational efficiency of different implementations of ASCBS on instances with varying cycle times, with a fixed number of agent streams at $7$. We use average running time and success rate as performance metrics.

Figure \ref{fig:exp2} presents the result. As the cycle time increases, the success rate tends to rise while the running time decreases. This trend is due to the sparser distribution of agents and the resulting ease in finding solutions with longer cycle times. In all graphs of the experiments, the performance of ASCBS-A-ND and ASCBS-A-D outperforms that of ASCBS-IDA-ND and ASCBS-IDA-D, indicating that resolving conflicts within the same agent stream in the high-level solver consistently has better results than in the low-level solver. Recall that the cyclic vertex/edge constraint can be extracted to approximately $\lfloor \frac{l}{c} \rfloor$ vertex/edge constraints where $l$ is the path length and $c$ is the cycle time. In 'empty-8-8', the ASCBS-A-ND is better than the ASCBS-A-D due to the large $\lfloor \frac{l}{c} \rfloor$ with the cycle time $c = 1$. However, in the other three maps, finding solutions with a cycle time of $1$ is challenging, resulting in low success rates for all algorithms. In 'random-64-64-10', ASCBS-A-D slightly outperforms ASCBS-A-ND when $c$ exceeds $2$, since the small path length $l$ in this map leads to a small value of $\lfloor \frac{l}{c} \rfloor$. Thus, the advantage of the disjoint splitting strategy is larger than the advantage of the cyclic vertex/edge constraints. Conversely, in 'Paris\_1\_256', the ASCBS-A-ND is better than the ASCBS-A-D. This is because the $l$ is relatively large, making  $\lfloor \frac{l}{c} \rfloor$ large and the advantage of the disjoint splitting is more significant than the advantage of the cyclic vertex/edge constraints.

\section{Disucssion on the Strategy }
Let $\langle as_i, as_j, q_i, q_j, v \rangle$ represent a cyclic vertex constraint.
For the ASCBS without disjoint splitting, when $as_i = as_j$,  it is straightforward to generate two child CT nodes:
\begin{itemize}
    \item Add cyclic vertex constraint $\overline{(as_i, v, q_i, q_i)}$ to the constraint set and replan the path of $as_i$.
    \item Add cyclic vertex constraint $\overline{(as_j, v, q_j, q_j)}$ in the constraint set and replan the path of $as_j$.
\end{itemize}
However, this strategy does not uphold Lemma \ref{lemma:2}, as there may be a collision-free path for $as_i$ and an integer $q$ such that the path satisfies the parent CT node's constraint set and it is located at vertex $v$ at $q^{th}$ step and $q \equiv q_i \ (mod \ c)$, $q \neq q_i$, $q \neq q_j$. However, this path cannot satisfy the cyclic vertex constraints $\overline{(as_i, v, q_i, q_i)}$ and $\overline{(as_j, v, q_j, q_j)}$.

Furthermore, the completeness of the algorithm can not be preserved because the growth process of CT might lose some collision-free paths without the guarantee of Lemma \ref{lemma:2}.

\begin{figure*}[!t]
    \centering

    \includegraphics[width=0.98\textwidth]{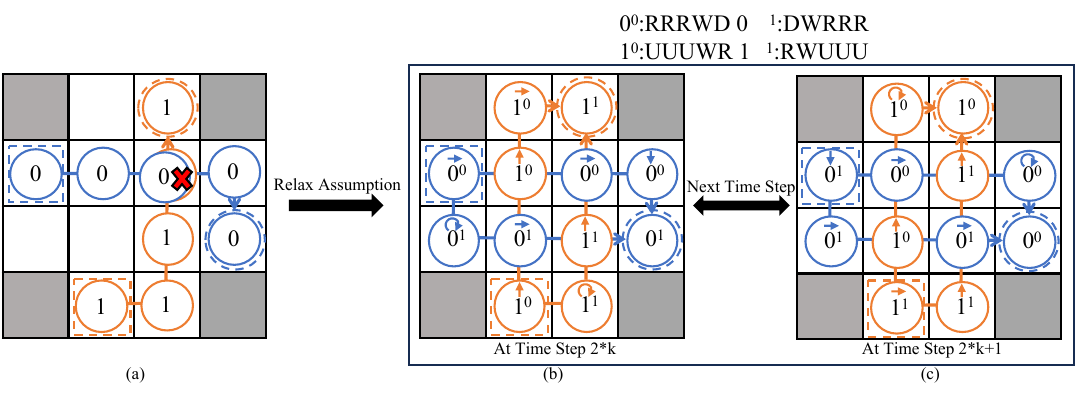}

    \hfill
   
    \caption{An instance whose cycle time is $1$.
    (a) There is no feasible solution with the assumption that the agent in the same agent stream should share the same action sequence, while (b) and (c) can form a solution without this assumption. (b) and (c) are the snapshots at the step time $2 \times  k$ and $2 \times k + 1$, respectively, where $k$ is where $k$ is a sufficiently large integer. The agents in the agent stream are divided into two groups based on the odd or even time and the superscript of the agent stream id indicates the group number. The arrow inside the circle demonstrates the next action of the corresponding agent. The action sequence is shown at the top of the figure. }
    \label{fig:relaxtion}
\end{figure*}

\section{Discussion on the Same Action Sequence Assumption}
In this section, we will show that the assumption that agents in the same agent stream should have the same action sequence hurts the solvability of the instance. An instance is considered solvable if and only if it has a feasible solution. Despite its impact on solvability, this assumption is still deemed valuable. We use $A_F$ to denote the assumption.
\begin{theorem}
Let $\Pi_0$ denote the set of solvable S-MAPF problem instances and $\Pi_1$ denote the set of solvable instances that is the S-MAPF problem without the assumption that agents in the same agent stream should have the same action sequence. It follows that $\Pi_0 \subset \Pi_1$.
\end{theorem}
\begin{proof}
It is evident that $\Pi_0 \subseteq \Pi_1$, since $\Pi_1$ has fewer constraints and the instance in $\Pi_0$ must be in $\Pi_1$. In the following, we will demonstrate an instance that is in $\Pi_1$, but not in $\Pi_0$. Figure \ref{fig:relaxtion} shows an S-MAPF instance with circle time $1$. With the assumption of $A_F$, an agent must appear in the start vertex for each time step and cannot perform a wait action during execution, because the wait action will cause the agent to collide with the next agent in the same agent stream. As shown in Figure \ref{fig:relaxtion}(a), the path of the agent stream $0$ will separate the graph into two isolated parts such that the start and the goal vertex of the agent stream $1$ are located on each one. This makes it impossible to plan a path for agent stream $1$ without collision. Therefore, the instance is not in $\Pi_0$. However, we can find a feasible solution without the assumption of $A_F$, as shown in Figures \ref{fig:relaxtion}(b) and \ref{fig:relaxtion}(c). For each agent stream, agents are divided into two groups based on the odd or even time steps in which they appear. Each agent group has its own path and shares a sequence of actions. The solution is collision-free, but not under the assumption of $A_F$, which puts this instance in $\Pi_1$.
\end{proof}
Despite its impact on solvability, adopting assumption $A_F$ is beneficial.
On the one hand, $A_F$ reduces the joint action space of the agent in the same agent stream and makes it possible to apply it in an unlimited working time. 
In Subsection 5.2 of the main text, an experiment is conducted to compare the computational efficiency under a time horizon between the solvers with assumption $A_F$ (ASCBS-A-ND) and without assumption $A_F$ (CBS). It shows the solver without the constraint on the joint action space is unavailable when the time horizon is large. On the other hand, the probability of the unsolvable instance sharply declines with increasing circle time and map size. Among all successful instances in our experiment, the instances that can be solved within 60 seconds but lack a feasible solution only exist in the 'empty-8-8' when the circle time is $1$. In conclusion, the S-MAPF problem is a practical model for navigating in assembly line scenarios.

\section{Discussion of Non-Uniform Cycle Time }
In this section, we present the extended definitions of cyclic conflict and constraints for the ASCBS algorithm with non-uniform cycle time. We only introduce cyclic vertex conflicts and constraints, omitting cyclic edge conflicts and constraints because of their symmetric nature.

The agent stream $i$ can be described as $as_i=\{v^s_i, v^g_i, t^s_i, c_i\}$ where $v^s_i$ and $v^g_i$ are the start and goal vertex, respectively. The $c_i$ is the cycle time of $as_i$ and the initial start time of $as_i$ is denoted by $t^s_i \in [0, c_i-1]$. It means that at each time step $k * c_i + t^s_i$ for all $k \in \mathbb{N}$, an agent of $as_i$ will appear at $v^s_i$. The solution can be denoted as $P = \{p_0, p_1,...,p_{n-1}\}$, and the path of $as_i$ is $p_i=[p_i^0, p_i^1, ..., p_i^{l_i-1}]$, with $l_i$ as the path length. We name $p^j_i$ the $j$-th \textit{step} of agent stream $as_i$.

The cyclic vertex conflict is denoted as $\langle as_i, as_j, q_i, q_j, v, c_i, c_j \rangle$. It occurs if and only if there exist two nonnegative integers $k_i$, $k_j$ such that the following equations are satisfied:
\begin{equation}
    p_i^{q_i} = v = p_j^{q_j}
\end{equation} 
\begin{equation} \label{equ:uv2}
    (k_i \cdot c_i + t^s_i) + q_i = (k_j \cdot c_j + t^s_j) + q_j
\end{equation}
\begin{equation} 
    (i - j)^2 + (k_i - k_j)^2 \neq 0
\end{equation}
It means that the agent of $as_i$ with start time $k_i \cdot c_i + t^s_i$ is located at the same vertex $v$ at the time step $(k_i \cdot c_i + t^s_i) + q_i$ with the agent of $as_j$ with start time $k_j \cdot c_j + t^s_j$. In addition, when $as_i$ and $as_j$ are the same agent stream ($j - i = 0$), $k_i$ equal to $k_j$ won't cause a conflict because these two agents are the same agent. Equation \ref{equ:uv2} and can be further derived:
\begin{equation} \label{equ:uv2n}
    k_i \cdot c_i - k_j \cdot c_j = t^s_j + q_j - t^s_i - q_i
\end{equation}
Let $g$ be the greatest common divisor of $c_i$ and $c_j$. If $t^s_j + q_j - t^s_i - q_i$ is not a multiplier of $g$, then Equation \ref{equ:uv2n} must not be satisfied, because the $k_i \cdot c_i - k_j \cdot c_j$ is a multiplier of $g$. Conversely, if $t^s_j + q_j - t^s_i - q_i$ is a multiplier of $g$, there must exist two nonnegative integers $k_i$, $k_j$ such that Equation \ref{equ:uv2n} is satisfied, according to the Lemma of Bézout \cite{mcconnell2011certifying} in number theory.
In the implementation of the ASCBS algorithm, for each vertex, we can record all agent streams that have passed through and check for potential cyclic vertex conflicts.

The cyclic vertex constraint is denoted as $\overline{(as, v, q_r, q_e, c_i)}$, indicating that the agent stream $as$ cannot occupy the vertex $v$ at the $q$-th step if $q$ satisfies the equations:
\begin{equation} 
     q \equiv q_r \ (mod \ c_i) 
\end{equation}
\begin{equation} 
     q \neq q_e
\end{equation}
The notation $\overline{(as, v, q)}$ represents the vertex constraint, indicating that the agent stream $as$ cannot be located at the vertex $v$ at the $q$-th step.

The conflict resolution process is similar to the algorithm in the main text. Here, we focus on the ASCBS algorithm without the disjoint splitting strategy. To resolve the vertex conflict between different agent streams, the two child CT nodes are generated respectively as follows to resolve the cyclic vertex conflict $\langle as_i, as_j, q_i, q_j, v, c_i, c_j \rangle$ where $as_i \neq as_j$:
\begin{itemize}
    \item Add the cyclic vertex constraint $\overline{(as_i, v, q_i, \emptyset, c_i)}$ to the constraint set and replan the path of $as_i$.
    \item Add the cyclic vertex constraint $\overline{(as_j, v, q_j, \emptyset, c_j)}$ to the constraint set and replan the path of $as_j$.
\end{itemize}

 To resolve the cyclic vertex conflict $\langle as_i, as_j, q_i, q_j , v, c_i, c_j \rangle$ where $as_i = as_j$, two child CT nodes are generated:
\begin{itemize}
    \item Add the vertex constraint $\overline{(as_i, v, q_i)}$ to the constraint set and replan the path of $as_i$.
    \item Add the vertex constraint $\overline{(as_i, v, q_j)}$ to the constraint set and replan the path of $as_i$.
\end{itemize}

\end{document}